\documentclass[preprint,12pt]{article}

\usepackage{amssymb}
\usepackage{amsthm}
\usepackage{amsmath}
\usepackage{graphicx}

\begin{document}

\title{PI Consensus Error Transformation for Adaptive Cooperative Control of Nonlinear Multi-Agent Systems} 

\author{Haris E. Psillakis \\ National Technical University of Athens\\
H. Polytechniou 9, 15780 Athens, Greece\\
\texttt{hpsilakis@central.ntua.gr}}

\date{}

\maketitle

\begin{abstract}
A solution is provided in this note for the adaptive consensus problem of nonlinear multi-agent systems with unknown and non-identical control directions assuming a strongly connected underlying graph topology.
This is achieved  with the introduction of a novel variable transformation called PI consensus error transformation. 
The new variables include the position error of each agent from some arbitrary fixed point along with an integral term of the weighted total displacement of the agent's position from all neighbor positions. It is proven that if these new variables are bounded and regulated to zero, then asymptotic consensus among all agents is ensured. The important feature of this transformation is that it provides input decoupling in the dynamics of the new error variables making the consensus control design a simple and direct task.
Using  classical Nussbaum gain based techniques, distributed controllers  are  designed  to regulate the PI consensus error variables to zero and ultimately solve the agreement problem. The proposed approach also allows for a specific calculation of the final consensus point based on the controller parameter selection and the associated graph topology. Simulation results verify our theoretical derivations.
\end{abstract}

\newtheorem{corollary}{Corollary}
\newtheorem{theorem}{Theorem}
\newtheorem{assumption}{Assumption}
\newtheorem{lemma}{Lemma}
\newtheorem{remark}{Remark}
\newtheorem{definition}{Definition}

\section{Introduction}
%
%
%
%
Cooperative control has received growing interest among researchers over the last decades with applications in several areas including
spacecraft formation flying, sensor networks, and cooperative
surveillance \cite{OlfatiFaxMurray2007}, \cite{RenBeardAtkins2007}.  Consensus  is an important and fundamental problem within the field of cooperative control aiming to design distributed control algorithms using only local information to ensure that the agents reach an agreement on certain variables of interest \cite{JadbabaieLinMorse03}, \cite{OlfatiMurray04}. Detailed literature reviews on the subject can be found in the works \cite{RenBeardBook08,ShammaBook08,MehranEgerstedtBook10,RenCaoBook11,LewisBook2014}.

In some applications \cite{du2007adaptiveieeejoe}, \cite{AstolfiVisualServoing} the control directions might not be known a priori. R. Nussbaum in the seminal paper \cite{nussbaum1983somescl} proposed a class of nonlinear control gains to resolve this issue. Since then, the so called Nussbaum gain technique has been successfully applied and generalized  to different system classes \cite{ye98,ding1998globalieeeac,zhang2000adaptiveieeeac,ZPJiang04,liu2006globalieeeac, Psillakis2010, yan2010globalsiamsiamjco}.

The last few years the cooperative control of multi-agent systems with unknown control directions has been under research.  Chen \emph{et al.} \cite{chen2014adaptiveieeeac} considered the adaptive consensus of first-order and second-order agents with unknown identical control directions using a novel Nussbaum function. The consensus control problem with unknown identical control directions was also addressed later in \cite{radenkovic2016multi} using  delayed inputs and switching functions, in \cite{ma2017cooperativeamc} for higher-order integrators and in \cite{wang2017nonlinear} for nonlinear systems. The  unknown identical control directions assumption was relaxed in \cite{shi2015cooperativeieeease} and \cite{wang2016prescribed} for agent networks having a leader. Nussbaum functions were also employed for unknown identical and non-identical control directions in \cite{ding2015adaptiveauto,liu2015adaptiveieeeac,su2015cooperativeieeeac,Guo2017regulationieeetac}  in the framework of cooperative output regulation with the exosystem having the leader role.

For leaderless networks, Peng and Ye \cite{peng2014cooperativescl} removed this assumption but only for single-integrator systems without nonlinearities.
Also, in \cite{psillakis2016consensusieeeac},  the same problem was investigated under switching topologies using the nonlinear PI method \cite{ortega2002nonlinearscl,psillakis2016integratorejc}.
Recently in \cite{chen2016adaptiveieeeac}, Chen \emph{et al.} generalized the results of \cite{chen2014adaptiveieeeac} to nonlinear systems with partially unknown control directions using a novel Nussbaum function. However, as stated in the conclusion of \cite{chen2016adaptiveieeeac} the adaptive consensus problem with completely unknown non-identical control directions by using the Nussbaum function based approach is still an interesting open issue for further research.

In this paper we provide a solution to the above problem using a different line of attack. We introduce a state transformation into the so called \emph{PI consensus error variables} that include the position error of each agent from some arbitrary fixed point  along with an integral term of the weighted total displacement between the agent's position and all neighbor positions (consensus error). It is proven that if the new variables are bounded and converge towards zero then asymptotic consensus among all agents is ensured. The important feature of this transformation is that it provides input decoupling in the dynamics of the new error variables making the consensus control design a simple and direct task. Applying the proposed transformation and using standard Nussbaum gain techniques we obtain a  straightforward solution for the unsolved problem of adaptive consensus for nonlinear systems with unknown and non-identical control directions.

The rest of this paper is organized as follows. In Section \ref{section2}, some preliminaries on graph theory and  Nussbaum functions  are recalled. Also, the adaptive consensus design problem of this work is formulated. In Section \ref{section3}, we propose the new PI consensus error transformation and prove the key lemma motivating its use. Adaptive consensus designs are given for first and second-order nonlinear agent models with unknown and non-identical control directions in Section \ref{section4}. In  Section  \ref{section5}, the obtained results are verified by a simulation example. Finally,  some concluding remarks are given in Section \ref{section6}.

%
%
%
%
%
\section{Preliminaries}\label{section2}
\subsection{Graph Theory}
In what follows, we revisit basic definitions from graph theory \cite{RenBeardBook08,ShammaBook08,MehranEgerstedtBook10,RenCaoBook11}.  A
directed graph is denoted by $\mathcal{G}=\left( \mathcal{V},\mathcal{E},\mathcal{A}\right)$ where $\mathcal{V}=\left\{ v_{1},v_{2},\ldots ,v_{N}\right\}
$ represents the nonempty set of nodes and $\mathcal{E}\subseteq \mathcal{%
V}\times \mathcal{V}$ is the set of edges. Matrix $\mathcal{A}=\left[ a_{ij}\right] \in\mathbb{R}
^{N\times N}$ is called the adjacency matrix and the element $a_{ij}$ represents the
coupling strength of edge $\left( j,i\right)$ with $a_{ij}>0$ if $\left( j,i\right)\in\mathcal{E}$  and $a_{ij}=0$ otherwise. 
The set of neighbors of $i$-th agent is then defined by $\mathcal{N}_{i}=\left\{ j\in \mathcal{V}:\left( j,i\right) \in
\mathcal{E}\right\} $.
Let $d_{i}=\sum_{j=1}^{N}a_{ij}$ be the in-degree of vertex $i$, and denote $\mathcal{D}=\mathrm{%
diag}\left\{ d_{1},\ldots ,d_{N}\right\} $ the in-degree matrix. Then the Laplacian matrix is defined as $L=\mathcal{D}-\mathcal{A}$. The directed path with length $l$ is defined with a sequence of edges in the form $\left( \left( i_{1},i_{2}\right) ,\left(
i_{2},i_{3}\right) ,\ldots ,\left( i_{l},i_{l+1}\right) \right) $ where $%
\left( i_{j},i_{j+1}\right) \in \mathcal{E}$ for $j=1,\ldots ,l$ and $%
i_{j}\neq i_{k}$ for $j,k=1,\ldots ,l$ and $j\neq k$.
If there exists a directed path between any two distinct nodes in directed graph $\mathcal{G}$, the graph is said to be strongly
connected.

A strongly connected graph has a simple zero eigenvalue associated with right eigenvector $\mathbf{1}:=[1,1,\cdots,1]^T$ and left eigenvector $\omega:=[\omega_1, \omega_2,\cdots,\omega_N]^T$ where $\omega_i$ are positive real numbers (Corollary 3.2 of \cite{LiJiaIJC2009}) with $\sum_{i=1}^N{\omega_i}=1$. All other eigenvalues of $L$ have positive real parts for a strongly connected graph. Thus, $L$ admits  the following Jordan decomposition
\begin{equation}\label{L_jordan_decomp}
  L=\left[\begin{array}{cc}
                                                                              U & \mathbf{1}
                                                                            \end{array}
\right]\left[
         \begin{array}{cc}
           D & 0 \\
           0 & 0 \\
         \end{array}
       \right]\left[
                \begin{array}{c}
                  V^T \\
                  \omega^T \\
                \end{array}
              \right]=UDV^T
\end{equation}
where $D$ is a $(N-1)\times(N-1)$ upper triangular matrix that includes all Jordan blocks related to the $N-1$ eigenvalues of $L$ with positive real part. $U, V\in \mathbb{C}^{N\times (N-1)}$ are  matrices with columns the right and left eigenvectors related to these eigenvalues respectively. Then, for the exponential matrix $e^{Lt}$ it holds true that
\begin{equation}\label{expmLaplacian}
  e^{Lt}=Ue^{Dt}V^T+\mathbf{1}\omega^T.
\end{equation}
\subsection{Nussbaum functions}
\begin{definition} \cite{nussbaum1983somescl}
\label{Nussbaumtypefunc}The function $N\left( \cdot \right) $ is called a
Nussbaum-type function if it has the following properties:%
\begin{equation}
\left\{
\begin{array}{c}
\lim_{k\rightarrow \infty }\sup \left( \frac{1}{k}\int_{0}^{k}N\left( \tau
\right) d\tau \right) =+\infty \\
\lim_{k\rightarrow \infty }\inf \left( \frac{1}{k}\int_{0}^{k}N\left( \tau
\right) d\tau \right) =-\infty%
\end{array}%
\right.   \label{nussbaumtypedefine}
\end{equation}
\end{definition}

Commonly used Nussbaum-type functions include $e^{k^{2}}\cos(\pi k/2)$, $k^{2}\sin(k)$ and $k^{2}\cos(k)$ among others. In this paper, we choose an even smooth Nussbaum function $N(k)=k^{2}\cos(k)$. The following lemma is central in the analysis of Nussbaum control schemes.
\begin{lemma}\label{lemma_Nuss_convergence}\cite{ye98}
Let $V(\cdot)$ and $k(\cdot)$ be smooth functions defined on $[0,t_f)$ with $V (t)\geq 0$  $\forall t \in [0, t_f)$, $N(\cdot)$ be an even smooth Nussbaum-type function, and $\theta_0$ be a nonzero constant. If the
following inequality holds:
\begin{equation}\label{Nussbaum_convergence}
  V (t) \leq \int_0^t{(\theta_0N(k(\tau))+1)\dot{k}(\tau)d\tau}+c, \quad \forall t\in[0,t_f)
\end{equation}
where $c\in\mathbb{R}$ is a constant, then $V (t)$, $k(t)$ and $\int_0^t{(\theta_0N(k(\tau))+1)\dot{k}(\tau)d\tau}$ must be bounded on $[0,t_f)$.
\end{lemma}
\subsection{Problem formulation}

The nonlinear agent models of \cite{chen2014adaptiveieeeac} are considered in this work. We assume either $N$ first-order agents with state
$x_{i}\in\mathbb{R}$ and dynamics
\begin{equation}
\dot{x}_{i}=b_{i}u_{i}+\theta_i^T\phi_i(x_i),\quad i=1,2,\dots ,N,  \label{agentdynamicssi}
\end{equation}%
or $N$ second-order agents with position $x_{i}\in\mathbb{R}$, velocity $v_{i}\in\mathbb{R}$ and dynamics
\begin{equation}
\left\{
\begin{array}{l}
\dot{x}_{i}=v_{i} \\
\overset{\cdot }{v}_{i}=b_{i}u_{i}+\theta_i^T\phi_i(x_i,v_i)%
\end{array}%
\right. ,\quad i=1,2,\dots ,N,  \label{agentdynamicsdi}
\end{equation}
where $u_{i}\in\mathbb{R}$ is the control input and $b_{i}\in\mathbb{R}$ is the high frequency gain. Function $\phi_i:\mathbb{R}\rightarrow\mathbb{R}^{\ell_i}$ for first-order agent model ($\phi_i:\mathbb{R}\times\mathbb{R}\rightarrow\mathbb{R}^{\ell_i}$ for second-order agent model) is a known regressor vector of agent smooth nonlinearities and $\theta_i\in \mathbb{R}^{\ell_i}$ is an unknown parameter vector.

\begin{assumption}
\label{assumcontrolgain}The control gains $b_{i},$ $i=1,2,\ldots ,N$ are
unknown and nonzero constants.
\end{assumption}

\begin{remark}
The assumption  $b_{i}\neq 0$ for all $i=1,2,\ldots ,N$ is necessary for the controllability of each agent dynamics. The signs of the gains $b_i$ may be different and their prior knowledge is not needed.
\end{remark}
\begin{assumption}
\label{assumption_topology}The underlying graph topology is strongly connected.
\end{assumption}

The control objective is to design a new class of algorithms for agents \eqref{agentdynamicssi} or \eqref{agentdynamicsdi} under Assumptions \ref{assumcontrolgain} and \ref{assumption_topology} to achieve consensus such that
\begin{equation}
\lim_{t\rightarrow \infty }\left( x_{i}(t)-x_{k}(t)\right) =0.
\label{objectivesi}
\end{equation}
for first-order agents with $i,k\in\{1,2,\cdots,N\}$ or
\begin{equation}
\left\{
\begin{array}{c}
\lim_{t\rightarrow \infty }\left( x_{i}(t)-x_{k}(t)\right) =0 \\
\lim_{t\rightarrow \infty }\left( v_{i}(t)-v_{k}(t)\right) =0%
\end{array}%
\right.   \label{objectivedi}
\end{equation}%
for second-order agents with $i,k\in \left\{ 1,2,\ldots ,N\right\} $.

\section{PI Consensus Error Transformation}\label{section3}
Let $x:=[x_1,x_2,\cdots,x_N]^T\in\mathbb{R}^N$ the agents position vector and assume that $L$ is the Laplacian matrix associated with the underlying graph topology. We define the new variables
\begin{equation}\label{z_definition}
  z_i(t):=x_i(t)-\bar{x}_i+\rho\int_0^t{\sum_{j\in\mathcal{N}_i}a_{ij}(x_i(s)-x_j(s))ds}
\end{equation}
for $i=1,2,\cdots, N$ with $\rho>0$ and $\bar{x}_i\in\mathbb{R}$ some arbitrary fixed point. Define also the vectors $\bar{x}:=[\bar{x}_1,\bar{x}_2,\cdots,\bar{x}_N]^T\in\mathbb{R}^N$  and $z:=[z_1,z_2,\cdots,z_N]^T\in\mathbb{R}^N$. The above transformation is called \emph{PI consensus error transformation} since it includes the position error $x_i-\bar{x}_i$ of each agent from some arbitrary fixed point  along with an integral term of the weighted total displacement of the agent's position from all neighbor positions $\sum_{j\in\mathcal{N}_i}{a_{ij}(x_i-x_j)}$ (consensus error). The following lemma holds true which motivates the use of the new variables $z_i$ in the design procedure.
\begin{lemma}\label{main_lemma}
Consider a set of $N$ agents with positions $x_i:[0,\infty)\rightarrow\mathbb{R}$. If \begin{description}
     \item[i)] the underlying graph is strongly connected
     \item[ii)] $x_i(t)$ is continuous for all $i=1,2,\cdots,N$
     \item[iii)] $z_i$ is bounded with $\lim_{t\rightarrow\infty}z_i(t)=0$ for all $i=1,2,\cdots,N$
   \end{description}
then $\lim_{t\rightarrow\infty}(x_i(t)-x_j(t))=0$ and $\lim_{t\rightarrow\infty}x_i(t)=\sum_{j=1}^N{\omega_j\bar{x}_j}$
where $\omega:=[\omega_1,\cdots,\omega_N]^T\in\mathbb{R}^N$ is the left eigenvector of $L$ associated with the zero eigenvalue.
\end{lemma}
\begin{proof}
Define
\begin{equation*}
  w_i(t):=\int_0^t{\sum_{j\in\mathcal{N}_i}a_{ij}(x_i(\tau)-x_j(\tau))d\tau}
\end{equation*}
and the vector variable $w(t):=[w_1(t),\cdots,w_N(t)]^T$. Variable $w(t)$ can also be written as
\begin{equation}\label{w_definition}
  w(t):=\int_0^t{Lx(s)ds}.
\end{equation}
Then, it holds true that
\begin{equation}\label{w_dynamics}
  \dot{w}(t)=L{x}(t)=-\rho L w(t)+ L(z(t)+\bar{x}).
\end{equation}
The solution vector of \eqref{w_dynamics} is $w(t)=f(t)+g(t)$ with
\begin{align}
  f(t)&:=\int_0^t{e^{-\rho L(t-s)}Lz(s)ds}\label{f_def}\\
  g(t)&:=\int_0^t{e^{-\rho L(t-s)}L\bar{x}ds}\label{g_def}.
\end{align}
Using the Jordan decomposition of the Laplacian matrix $L$ described in \eqref{L_jordan_decomp} and the exponential matrix \eqref{expmLaplacian} we result in $e^{-\rho Lt}L=Ue^{-\rho Dt}DV^T$. 
Thus, it holds true that
\begin{equation}\label{}
 g(t)=\int_0^t{Ue^{-\rho D(t-s)}DV^T\bar{x}ds}=\frac{1}{\rho}U\left(\mathbb{I}- e^{-\rho Dt}\right)V^T\bar{x}
\end{equation}
and therefore $\lim_{t\rightarrow\infty}g(t)=(1/\rho)UV^T\bar{x}$. For the other term $f(t)$ of $w(t)$ convergence to zero is also ensured if $z$ is bounded and $\lim_{t\rightarrow\infty}z(t)=0$. Since all diagonal elements of $D$ have positive real parts there exist some constants $\rho_1,\lambda_1>0$ such that $\|Ue^{-\rho Dt}DV^T\|\leq \lambda_1e^{-\rho_1 t}$ for all $t\geq 0$. Also, since $\lim_{t\rightarrow\infty}z(t)=0$ for every $\epsilon>0$ there exists time $T(\epsilon)>0$ such that $\|z(t)\|\leq \rho_1\epsilon/2\lambda_1$ for all $t\geq T(\epsilon)$. Moreover, the boundedness of $z$ implies that there exist some $c>0$ such that $\|z(t)\|\leq c$ for all $t\geq 0$. Thus, for time $t\geq \max\{2T(\epsilon),(2/\rho_1)\ln(2c\lambda_1/\rho_1\epsilon)\}$  it holds true that
\begin{align}\label{}
  \|f(t)\| \leq &\int_{0}^{t/2}{\|e^{-\rho L(t-s)}L \|\|z(s)\|ds}+\int_{t/2}^{t}{\|e^{-\rho L(t-s)}L \|\|z(s)\|ds}\nonumber\\
  \leq &  c\int_{0}^{t/2}{\|Ue^{-\rho D(t-s)}DV^T \|ds} +\frac{\rho_1\epsilon}{2\lambda_1}\int_{t/2}^{t}{\|Ue^{-\rho D(t-s)}DV^T \|ds}\nonumber\\
  \leq & \frac{c\lambda_1}{\rho_1}e^{-\rho_1t/2}+\frac{\epsilon}{2}\leq \epsilon.
\end{align}
Hence,  $\lim_{t\rightarrow\infty}f(t)=0$. From the above analysis,  $\lim_{t\rightarrow\infty}w(t)=\lim_{t\rightarrow\infty}[f(t)+g(t)]=(1/\rho)UV^T\bar{x}$   and since $x(t)=z(t)-\rho w(t)+\bar{x}$ the variable vector $x$ also converges to
\begin{align*}
 \lim_{t\rightarrow\infty}x(t)=&\lim_{t\rightarrow\infty}z(t)-\rho\lim_{t\rightarrow\infty} w(t)+\bar{x}\\
 =&\bar{x}-UV^T\bar{x}=\mathbf{1}\omega^T\bar{x}.
\end{align*}
The above limit yields the desired consensus property since  for the difference $x_i-x_j$ we have
\begin{equation}\label{xi-xj}
 \lim_{t\rightarrow\infty} (x_i(t)-x_j(t))=(e_i-e_j)^T\lim_{t\rightarrow\infty}x(t)=0
\end{equation}
and $\lim_{t\rightarrow\infty}x_i(t)=e_i^T\lim_{t\rightarrow\infty}x(t)=\omega^T\bar{x}$
where $e_i$ is the $i$-th column of the identity matrix.
\end{proof}
\begin{remark}\label{remark_xbar_selection}
If each agent selects  the fixed point $\bar{x}_i=x_i(0)$ in \eqref{z_definition} and the conditions of Lemma \ref{main_lemma} are true then all agents converge towards the weighted average of their initial conditions $\sum_{j=1}^N{\omega_jx_j(0)}$.
\end{remark}
\begin{remark}\label{remark_consensus_error_variables_usage}
It is to be noted that there are other error variables  such as $\xi_i:=\sum_{j\in\mathcal{N}_i}{a_{ij}(x_i-x_j)}$ which, if regulated to zero, also ensure asymptotic consensus among the agents. However, transformation \eqref{z_definition} has an input decoupling property which is very important in distributed control design. Only the input $u_i$ is present in the dynamics of $z_i$ while all neighbour inputs $u_j$ ($j\in\mathcal{N}_i$) are left out. This decoupling property does not occur in the dynamics of $\xi_i$ and is the main source of difficulty in earlier designs \cite{chen2014adaptiveieeeac}, \cite{chen2016adaptiveieeeac}. From this point of view, the use of transformation \eqref{z_definition} could possibly find applications in several other leaderless consensus problems involving different or more general systems classes than \eqref{agentdynamicssi} or \eqref{agentdynamicsdi}.
\end{remark}

\section{Distributed Control Design}\label{section4}
The proposed transformation reduces the consensus problem to a simple distributed regulation problem. If each agent's input is selected such that the corresponding PI consensus error variable is bounded and regulated to zero then consensus among all agents will occur. This is shown in Subsection \ref{subsection_first_order_agent_design} and \ref{subsection_second_order_agent_design} for agents of the form \eqref{agentdynamicssi} and \eqref{agentdynamicsdi} respectively.

\subsection{First-Order Agents}\label{subsection_first_order_agent_design}
Consider $N$ agents of the form \eqref{agentdynamicssi} satisfying Assumptions \ref{assumcontrolgain}-\ref{assumption_topology}. The dynamics of the new error variables $z_i$ defined in \eqref{z_definition} are given by
\begin{equation}\label{zi_dynamics}
 \dot{z}_i=b_iu_i +\theta_i^T\phi_i(x_i)+\rho \sum_{j\in\mathcal{N}_i}{a_{ij}(x_i-x_j)}.
\end{equation}
Let now the nonnegative function $V_i=(1/2)z_i^2$. From  \eqref{zi_dynamics} the $V_i$ time derivative has the following form
\begin{equation}\label{Vidot}
  \dot{V}_i=z_i\left[b_iu_i +\theta_i^T\phi_i(x_i)+\rho \sum_{j\in\mathcal{N}_i}{a_{ij}(x_i-x_j)}\right].
\end{equation}
Consider also a parameter vector estimators $\hat{\theta}_i$ of $\theta_i$ with estimation error $\tilde{\theta}_i:=\hat{\theta}_i-\theta_i$. If we select the parameter adaptation law
\begin{equation}\label{hatthetasi}
\dot{\hat{\theta}}_i=\gamma_i\phi_i(x_i)z_i
\end{equation}
with $\gamma_i>0$ then  the time derivative of the nonnegative function $\bar{V}_i:=V_i+(1/2\gamma_i)\|\tilde{\theta}_i\|^2$ takes the form
\begin{equation}\label{barVidot}
  \dot{\bar{V}}_i=z_i\left[b_iu_i +\hat{\theta}_i^T\phi_i(x_i)+\rho \sum_{j\in\mathcal{N}_i}{a_{ij}(x_i-x_j)}\right].
\end{equation}
Selecting now the distributed control law
\begin{align}\label{uisi}
 u_i=N(\zeta_i)\left[\hat{\theta}_i^T\phi_i(x_i)+\rho\sum_{j\in\mathcal{N}_i}{a_{ij}(x_i-x_j)}+\nu z_i\right]
\end{align}
with $\nu>0$ and Nussbaum parameter update law
\begin{align}\label{zetaisi}
 \dot{\zeta}_i=\nu z_i^2+z_i\left[\hat{\theta}_i^T\phi_i(x_i)+\rho\sum_{j\in\mathcal{N}_i}{a_{ij}(x_i-x_j)}\right]
\end{align}
we obtain
\begin{align}\label{dotbarvisi}
\dot{\bar{V}}_i=-\nu z_i^2+\dot{\zeta}_i+b_i N(\zeta_i)\dot{\zeta}_i.
\end{align}
Integrating \eqref{dotbarvisi} over $[0,t]$ we result in
\begin{align}\label{barVibound}
  {\bar{V}}_i(t)={\bar{V}}_i(0)&-\nu \int_0^t{z_i^2(s)ds}\nonumber\\
  &+\int_0^t{(b_iN(\zeta_i(\tau))+1)\dot{\zeta}_i(\tau)d\tau}.
\end{align}
The dynamical system $\dot{x}_{ag}=f_{ag}(x_{ag})$ with augmented state vector $x_{ag}:=[x_1,w_1,\hat{\theta}_1,\zeta_1,\cdots,x_N,\allowbreak w_N,\allowbreak \hat{\theta}_N,\zeta_N]^T$ defined by \eqref{agentdynamicssi}, \eqref{hatthetasi}, \eqref{uisi}, \eqref{zetaisi} has a smooth locally Lipschitz map $f_{ag}$ and therefore a maximal solution exists over some interval $[0,t_f)$ \cite{morris1973differential}. From  \eqref{barVibound} and Lemma \ref{lemma_Nuss_convergence} we have that $z_i,\hat{\theta}_i,\zeta_i$, $\int_0^t{z_i^2(\tau)d\tau}$ are bounded in $[0,t_f)$. The boundedness of $z$ implies the boundedness of $w(t)$ since $w(t)=\int_0^t{Ue^{-\rho D (t-\tau)}DV^T[z(\tau)+\bar{x}]d\tau}$ and therefore $x=\bar{x}+z-\rho w$ is also bounded. Thus, the whole solution $x_{ag}$ is bounded and the final time can be extended to infinity ($t_f=\infty$) \cite{morris1973differential}. Due to \eqref{uisi} we also have $u_i\in\mathcal{L}_{\infty}$ and therefore \eqref{zi_dynamics} yields $\dot{z}_i\in\mathcal{L}_{\infty}$.   Combining this fact with $z_i\in\mathcal{L}_{\infty}\cap\mathcal{L}_{2}$ we result from Barbal\'at's Lemma  in $\lim_{t\rightarrow\infty}z_i(t)=0$ $\forall i=1,2,\cdots,N$. Since all conditions of Lemma \ref{main_lemma} hold true asymptotic consensus is achieved. Thus, the following theorem has been proved.
\begin{theorem}\label{theorem_consensus_si}
Consider a set of $N$ agents with dynamics described by \eqref{agentdynamicssi} satisfying Assumptions \ref{assumcontrolgain}-\ref{assumption_topology}. If we select the control input \eqref{uisi} with parameter estimation law \eqref{hatthetasi} and Nussbaum parameter update law \eqref{zetaisi} then all signals are bounded in the closed loop and $\lim_{t\rightarrow\infty}(x_i(t)-x_j(t))=0$ with $\lim_{t\rightarrow\infty}x_i(t)=\sum_{j=1}^N{\omega_j\bar{x}_j}$.
\end{theorem}

\subsection{Second-Order Agents}\label{subsection_second_order_agent_design}
For second-order agents \eqref{agentdynamicsdi} satisfying Assumptions \ref{assumcontrolgain}-\ref{assumption_topology} we define the new filtered error variables
\begin{align}\label{s_definition}
  s_i:=&\dot{z}_i+\lambda z_i\nonumber\\
  =&v_i+\rho\xi_i+\lambda\left(x_i-\bar{x}_i+\rho\int_0^t{\xi_i(s)ds}\right)
\end{align}
with $\lambda>0$ ($1\leq i\leq N$). Equivalently, we have
\begin{equation}\label{zi_solution_wrt_si}
  z_i(t)=e^{-\lambda t}z_i(0)+\int_0^t{e^{-\lambda (t-\tau)}s_i(\tau)d\tau}.
\end{equation}
Following a similar analysis to the proof of $\lim_{t\rightarrow\infty}f(t)=0$ in Lemma \ref{main_lemma} we can prove from \eqref{zi_solution_wrt_si} that if $s_i$ is bounded with $\lim_{t\rightarrow\infty}s_i(t)=0$  then $\lim_{t\rightarrow\infty}z_i(t)=\lim_{t\rightarrow\infty}\dot{z}_i(t)=0$. Also, if $s_i$ is bounded in $[0,t_f)$ then, $|z_i(t)|\leq |z_i(0)|+(1/\lambda)\sup_{\tau\in[0,t_f)}|s_i(\tau)|$ for all $t\in[0,t_f)$.

The dynamics of $s_i$ are given by
\begin{align}\label{si_dynamics}
  \dot{s}_i=b_iu_i+\theta_i^T \phi_i(x_i,v_i) +\lambda v_i+\rho \sum_{j\in\mathcal{N}_i}{a_{ij}(v_i-v_j)}+\lambda \rho\sum_{j\in\mathcal{N}_i}{a_{ij}(x_i-x_j)}.
\end{align}
Define now the nonnegative function $P_i=(1/2)s_i^2$. From \eqref{si_dynamics} the $P_i$ time derivative has the following form
\begin{align}\label{Pidot}
  \dot{P}_i=s_i\bigg[b_iu_i +\theta_i^T  \phi_i(x_i,v_i)+\lambda v_i+\rho \sum_{j\in\mathcal{N}_i}{a_{ij}(v_i-v_j)}+\lambda \rho\sum_{j\in\mathcal{N}_i}{a_{ij}(x_i-x_j)}\bigg].
\end{align}
We consider also parameter vector estimators $\hat{\theta}_i$ of $\theta_i$ with estimation error $\tilde{\theta}_i:=\hat{\theta}_i-\theta_i$. If we select the parameter adaptation law
\begin{equation}\label{hatthetadi}
\dot{\hat{\theta}}_i=\gamma_i\phi_i(x_i,v_i)s_i
\end{equation}
with $\gamma_i>0$ then  the time derivative of the nonnegative function $\bar{P}_i:=P_i+(1/2\gamma_i)\|\tilde{\theta}_i\|^2$ takes the form
\begin{align}\label{barPidot}
  \dot{\bar{P}}_i=s_i\bigg[b_iu_i +\hat{\theta}_i^T  \phi_i(x_i,v_i)+\lambda v_i+\rho \sum_{j\in\mathcal{N}_i}{a_{ij}(v_i-v_j)}+\lambda \rho\sum_{j\in\mathcal{N}_i}{a_{ij}(x_i-x_j)}\bigg].
\end{align}
Selecting now the distributed control law
\begin{align}\label{uidi}
 u_i=N(\zeta_i)\bigg[\hat{\theta}_i^T &\phi_i(x_i,v_i)+\lambda v_i+\rho\sum_{j\in\mathcal{N}_i}{a_{ij}(v_i-v_j)}\nonumber\\
 &+\nu s_i+\lambda\rho \sum_{j\in\mathcal{N}_i}{a_{ij}(x_i-x_j)}\bigg]
\end{align}
with  $\nu>0$ and  Nussbaum parameter update law
\begin{align}\label{zetaidi}
 \dot{\zeta}_i=\nu s_i^2+s_i\bigg[\hat{\theta}_i^T  \phi_i(x_i,v_i)+\lambda v_i+\rho\sum_{j\in\mathcal{N}_i}{a_{ij}(v_i-v_j)}+\rho\lambda \sum_{j\in\mathcal{N}_i}{a_{ij}(x_i-x_j)}\bigg]
\end{align}
we obtain
\begin{align}\label{dotbarPidi}
\dot{\bar{P}}_i=-\nu s_i^2+\dot{\zeta}_i+b_i N(\zeta_i)\dot{\zeta}_i.
\end{align}
Integrating \eqref{dotbarPidi} over $[0,t]$ we result in
\begin{align}\label{barPibound}
  {\bar{P}}_i(t)={\bar{P}}_i(0)&-\nu \int_0^t{s_i^2(s)ds}+\int_0^t{(b_iN(\zeta_i(\tau))+1)\dot{\zeta}_i(\tau)d\tau}.
\end{align}
The dynamical system $\dot{y}_{ag}=g_{ag}(y_{ag})$ with augmented state vector $y_{ag}:=[x_1,v_1,w_1,\hat{\theta}_1,\zeta_1,\cdots,\allowbreak x_N,\allowbreak v_N,w_N,\hat{\theta}_N,\zeta_N]^T$ defined by \eqref{agentdynamicsdi}, \eqref{hatthetadi}, \eqref{uidi}, \eqref{zetaidi} has a smooth locally Lipschitz map $g_{ag}$ and therefore a maximal solution exists over some interval $[0,t_f)$ \cite{morris1973differential}. From  \eqref{barPibound} and Lemma \ref{lemma_Nuss_convergence} we have that $s_i,\hat{\theta}_i,\zeta_i$, $\int_0^t{s_i^2(\tau)d\tau}$ are bounded in $[0,t_f)$. The boundedness of $s$ implies the boundedness of $z$. This, in turn yields the boundedness of $w(t)$ since $w(t)=\int_0^t{Ue^{-\rho D (t-\tau)}DV^T[z(\tau)+\bar{x}]d\tau}$ and therefore $x=\bar{x}+z-\rho w$ and $v=s-\rho Lx-\lambda z$ are also bounded. Thus, the whole solution $y_{ag}$ is bounded and the final time can be extended to infinity ($t_f=\infty$) \cite{morris1973differential}. Due to \eqref{uidi} we also have $u_i\in\mathcal{L}_{\infty}$ and therefore \eqref{si_dynamics} yields $\dot{s}_i\in\mathcal{L}_{\infty}$.   Combining this fact with $s_i\in\mathcal{L}_{\infty}\cap\mathcal{L}_{2}$ we result from Barbal\'at's Lemma  in $\lim_{t\rightarrow\infty}s_i(t)=0$ for all $i=1,\cdots,N$. This, in turn yields $\lim_{t\rightarrow\infty}z_i(t)=0$ for all $i=1,\cdots,N$. Since all conditions of Lemma \ref{main_lemma} hold true asymptotic consensus among all agent positions is achieved. Finally, we have that $\lim_{t\rightarrow\infty}v_i(t)=\lim_{t\rightarrow\infty}(s_i(t)-\rho\xi_i(t)-\lambda z_i(t))=0$. Thus, the following theorem has been proved.
\begin{theorem}\label{theorem_consensus_di}
Consider a set of $N$ agents with dynamics described by \eqref{agentdynamicsdi} satisfying Assumptions \ref{assumcontrolgain}-\ref{assumption_topology}. If we select the control input \eqref{uidi} with parameter estimation law \eqref{hatthetadi} and Nussbaum parameter update law \eqref{zetaidi}  then all signals  in the closed loop are bounded and $\lim_{t\rightarrow\infty}(x_i(t)-x_j(t))=\lim_{t\rightarrow\infty}(v_i(t)-v_j(t))=0$ with $\lim_{t\rightarrow\infty}x_i(t)=\sum_{j=1}^N{\omega_j\bar{x}_j}$, $\lim_{t\rightarrow\infty}v_i(t)=0$.
\end{theorem}
\begin{remark}
The proposed distributed control approach can also be generalized to other agent models such as those considered in\cite{chen2016adaptiveieeeac} using some of the tools described therein.  The purpose of this note is not to solve the adaptive consensus problem in its more general form but to propose an alternative approach that significantly simplifies the design and allows for several generalizations.
\end{remark}
\section{Simulation Example}\label{section5}
In this Section, we consider a group of four agents with first-order (Case 1) or second-order dynamics (Case 2) described by \eqref{agentdynamicssi}, \eqref{agentdynamicsdi} respectively and a strongly connected topology $\mathcal{G}$ depicted in Fig. \ref{figure_graph_simulation}.
\begin{figure}
\centering
\includegraphics[width=0.4\columnwidth, bb=200 365 280 435]{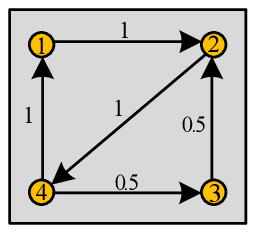}
\caption{Strongly connected unbalanced graph $\mathcal{G}$.}
\label{figure_graph_simulation}
\end{figure}
The left-eigenvector associated with the zero eigenvalue of the Laplacian matrix is $\omega=(1/9)[2,2,2,3]^T$.
In both cases completely unknown control directions are considered with $b_1=1$, $b_2=-2$, $b_3=2$, $b_4=-3/2$, parameters $\theta_1=\theta_2=1$, $\theta_3=-1$, $\theta_4=2$ and initial conditions $x(0)=[1,2,3,-1]^T$. For both cases the control and adaptation parameters $\rho=\nu=\gamma_i=0.1$ ($i=1,2,3,4$) and $\bar{x}=[1,2,3,4]^T$ are selected. For Case 1 the nonlinearities $\phi_1(x_1)=\sin(x_1)$, $\phi_2(x_2)=\cos(x_2^2)$, $\phi_3(x_3)=0.5x_3^2+1$, $\phi_4(x_4)=x_4\sin(x_4)$ are assumed. For Case 2 we consider $\phi_1(x_1,v_1)=\sin(x_1)\cos(v_1)$, $\phi_2(x_2,v_2)=v_2\cos(x_2^2)$, $\phi_3(x_3,v_3)=1+0.5x_3v_3$, $\phi_4(x_4,v_4)=\sin(x_4+v_4)$ and the additional parameter $\lambda=1.5$ and initial condition $v(0)=0$. Simulation results are shown in Fig. \ref{Fig:SIstate}-\ref{Fig:DIzeta}. As expected consensus is achieved in both cases with final consensus point $\lim_{t\rightarrow\infty}x_i(t)=\omega^T\bar{x}=8/3$ ($1\leq i\leq 4$) while all signals in the  closed-loop remain bounded.
\begin{figure}
\centering
\includegraphics[width=0.95\columnwidth]{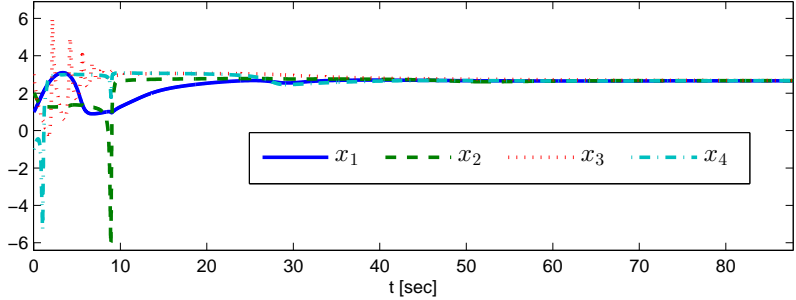}
\caption{The positions $x_{i}$ for first-order agents ($i=1,\ldots,4$).}
\label{Fig:SIstate}
\end{figure}
\begin{figure}[th]
\centering
\includegraphics[width=0.95\columnwidth]{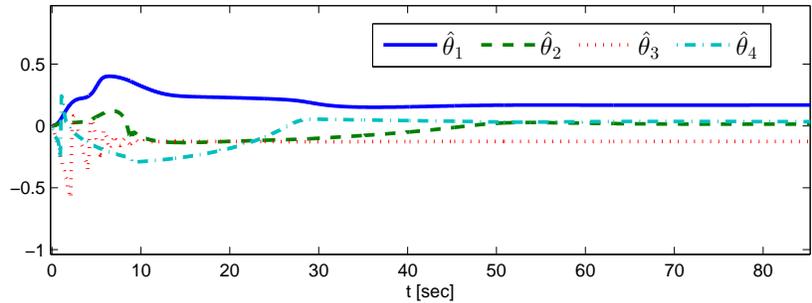}
\caption{The estimation variables  $\hat{\theta}_{i}$ for first-order agents ($i=1,\ldots,4$).}
\label{Fig:SItheta}
\end{figure}
\begin{figure}[th]
\centering
\includegraphics[width=0.95\columnwidth]{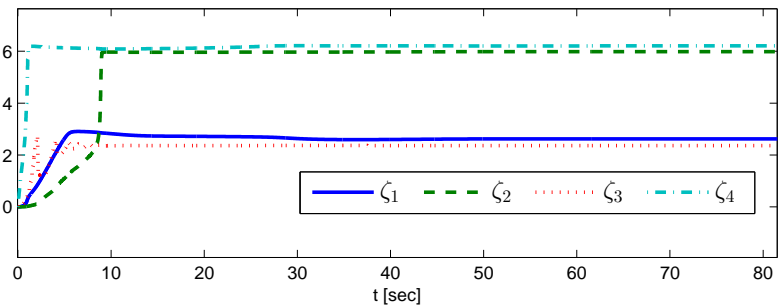}
\caption{The Nussbaum parameters  ${\zeta}_{i}$ for first-order agents ($i=1,\ldots,4$).}
\label{Fig:SIzeta}
\end{figure}
\begin{figure}[th]
\centering
\includegraphics[width=0.95\columnwidth]{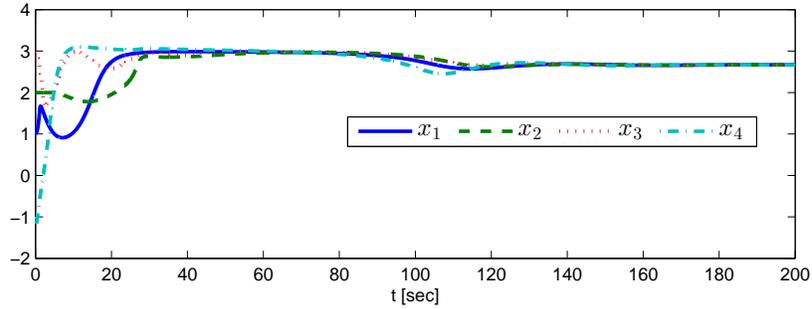}
\caption{The positions $x_{i}$ for second-order agents ($i=1,\ldots,4$).}
\label{Fig:DIstate}
\end{figure}
\begin{figure}[th]
\centering
\includegraphics[width=0.95\columnwidth]{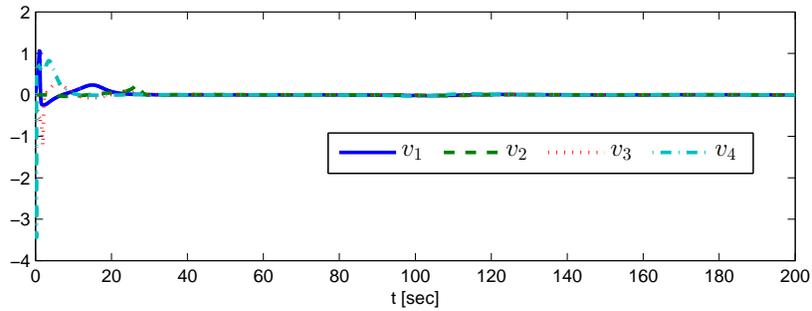}
\caption{The velocities $v_{i}$ for second-order agents ($i=1,\ldots,4$).}
\label{Fig:DIvel}
\end{figure}
\begin{figure}[th]
\centering
\includegraphics[width=0.95\columnwidth]{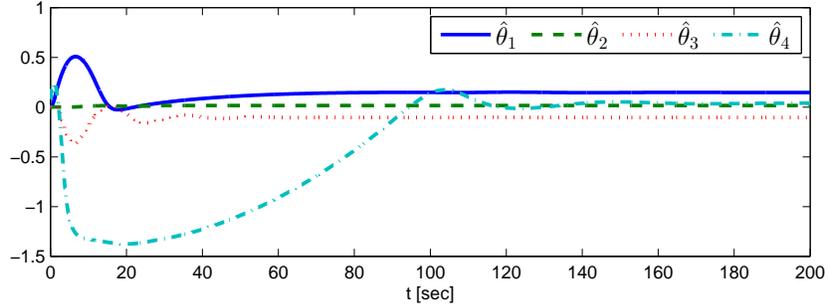}
\caption{The estimation variables  $\hat{\theta}_{i}$ for second-order agents ($i=1,\ldots,4$).}
\label{Fig:DItheta}
\end{figure}
\begin{figure}[th]
\centering
\includegraphics[width=0.95\columnwidth]{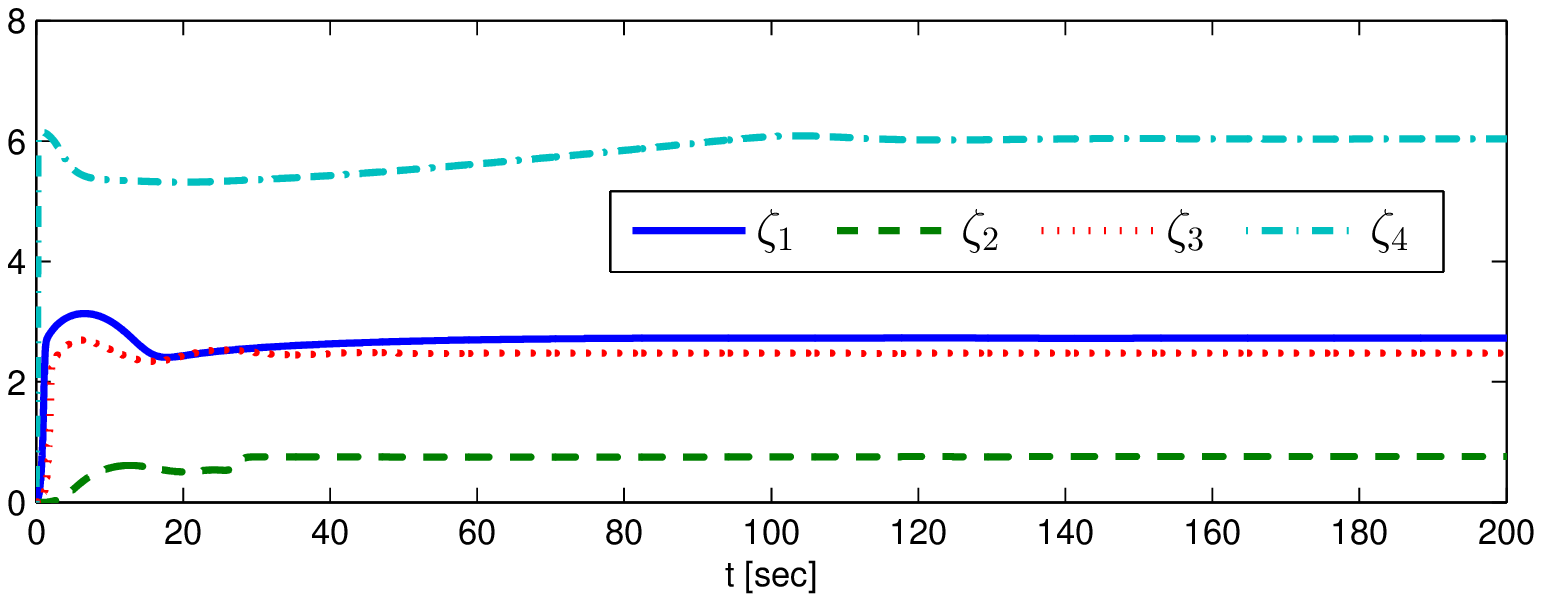}
\caption{The Nussbaum parameters  ${\zeta}_{i}$ for second-order agents ($i=1,\ldots,4$).}
\label{Fig:DIzeta}
\end{figure}

\section{Conclusion}\label{section6}
A novel change of variables is introduced in this work that transforms the consensus design problem into a simple regulation problem. Making use of this new transformation, an adaptive cooperative control law is proposed for nonlinear agents with unknown and non-identical control directions. Future work may explore  applications to more general system classes.

\end{document}